\documentclass[prl,twocolumn,twoside,superscriptaddress]{revtex4-1}

\pdfoutput=1

\usepackage{amsmath,amsthm,amssymb} %
\usepackage{caption} %
\usepackage{color} %
\usepackage{csquotes} %
\usepackage{dsfont} %
\usepackage{times} % charter, fourier, mathpazo, times
\usepackage{graphicx} %
\usepackage[breaklinks]{hyperref} %
\usepackage{subcaption} %
\usepackage{suffix} % for *-version commands
\usepackage{xspace} %

\hypersetup{colorlinks=true, linkcolor=blue, citecolor=magenta}

\newtheorem{theorem}{Theorem} %
\newtheorem{lemma}[theorem]{Lemma} %

\newcommand*{\ket}[1]{\ensuremath{|#1\rangle}} %
\newcommand*{\bra}[1]{\ensuremath{\langle#1|}} %
\newcommand*{\op}[2]{\ensuremath{\left|#1\right\rangle
    \left\langle#2\right|}} %
\newcommand*{\norm}[1]{\ensuremath{\left\lVert #1 \right\rVert}} %
\newcommand*{\complex}{\mathbb{C}} %
\newcommand*\dif{\mathop{}\nobreak \mskip-\thinmuskip\nobreak
  \mathrm{d}} %

\DeclareMathOperator{\Density}{D} %
\DeclareMathOperator{\Sep}{Sep} %
\DeclareMathOperator{\Unitary}{U} %

\DeclareMathOperator{\tr}{tr} %
\DeclareMathOperator{\Tr}{Tr} %

\def\H{\mathcal{H}} %

\begin{document}

%% End-Of-Header

\title{Detecting Consistency of Overlapping Quantum Marginals by
  Separability}

\author{Jianxin Chen} %
\affiliation{Joint Center for Quantum Information and Computer
  Science, University of Maryland, College Park, Maryland, USA}

\author{Zhengfeng Ji} %
\affiliation{Institute for Quantum Computing, University of Waterloo,
  Waterloo, Ontario, Canada} %
\affiliation{State Key Laboratory of Computer Science, Institute of
  Software, Chinese Academy of Sciences, Beijing, China}

\author{Nengkun Yu} %
\affiliation{Institute for Quantum Computing, University of Waterloo,
  Waterloo, Ontario, Canada} %
\affiliation{Department of Mathematics \& Statistics, University of
  Guelph, Guelph, Ontario, Canada}

\author{Bei Zeng} %
\affiliation{Institute for Quantum Computing, University of Waterloo,
  Waterloo, Ontario, Canada} %
\affiliation{Department of Mathematics \& Statistics, University of
  Guelph, Guelph, Ontario, Canada} %
\affiliation{Canadian Institute for Advanced Research, Toronto,
  Ontario, Canada}

\begin{abstract}
  The quantum marginal problem asks whether a set of given density
  matrices are consistent, i.e., whether they can be the reduced
  density matrices of a global quantum state. Not many non-trivial
  analytic necessary (or sufficient) conditions are known for the
  problem in general. We propose a method to detect consistency of
  overlapping quantum marginals by considering the separability of
  some derived states. Our method works well for the $k$-symmetric
  extension problem in general, and for the general overlapping
  marginal problems in some cases. Our work is, in some sense, the
  converse to the well-known $k$-symmetric extension criterion for
  separability.
\end{abstract}

% \date{\today}

\maketitle

The quantum marginal problem, also known as the consistency problem,
asks for the conditions under which there exists an $N$-particle
density matrix $\rho_N$ whose reduced density matrices (quantum
marginals) on the subsets of particles $S_i \subset \{1,2,\ldots, N\}$
equal to the given density matrices $\rho_{S_i}$ for all
$i$~\cite{Kly06}. The related problem in fermionic (bosonic) systems
is the so-called $N$-representability problem. It asks whether a
$k$-fermionic (bosonic) density matrix is the reduced density matrix
of some $N$-fermion (boson) state $\rho_N$. The $N$-representability
problem inherits a long history in quantum
chemistry~\cite{Col63,Erd72}.

The quantum marginal problem and the $N$-representability problem are
in general very difficult. They were shown to be the complete problems
of the complexity class QMA, even for the relatively simple case where
the given marginals are two-particle states~\cite{Liu06,LCV07,WMN10}.
In other words, even with the help of a quantum computer, it is very
unlikely that the quantum marginal problems can be solved efficiently
in the worst case. In this sense, the best hope to have simple
analytic conditions for the quantum marginal problem is to find either
necessary or sufficient conditions in certain special cases.

When the given marginals are states of non-overlapping subsets of
particles, and one is interested in a global \emph{pure} state
consistent with the given marginals, both the quantum marginal problem
and the $N$-representability problem were solved
~\cite{Kly04,Kly06,AK08,SGC13,walter2013entanglement,sawicki2014convexity}. However, not much is known for the
general problem with overlapping subsystems. For the tripartite case
of particles $A$, $B$, $C$, the strong subadditivity of von Neumann
entropy enforces non-trivial necessary conditions for the consistency
of $\rho_{AB}$ and $\rho_{AC}$ such as
$S(AB) + S(AC) \ge S(B) + S(C)$~\cite{CLL13}. In a similar spirit,
certain quantitative monogamy of entanglement type of results (see
e.g.~\cite{Coffman2000}) also put non-trivial necessary conditions.
Necessary and sufficient conditions are generally not known, except in
very few special situations~\cite{Smi65,chen2014symmetric,CLL13} when
$N$ is small.

In this work, we propose a simple but powerful analytic necessary
condition for arguably the simplest overlapping quantum marginal
problem, known as the $k$-symmetric extension problem. That is, we
will consider quantum marginal problems of $k+1$ particles $A$,
$B_1, B_2, \ldots, B_k$ for a given density matrix $\rho_{AB}$, and
require that there is a global quantum state
$\rho_{AB_1B_2\cdots B_k}$ whose marginals on $A, B_i$ equal to the
given $\rho_{AB}$ for $i=1,2,\ldots, k$. The classical analog of this
particular case is trivial and there is a consistent global
probability distribution as long as the marginals agree on $A$. In the
quantum case, however, the problem remains unsolved even for $k=2$.

We prove the separability of certain derived state as a necessary
condition for the $k$-symmetric extension problem. A quantum state
$\rho_{AB}$ is separable if it can be written as the convex
combination $\sum_i p_i \rho_{A,i}\otimes \rho_{B,i}$ for a
probability distribution $p_i$ and states $\rho_{A,i}$ and
$\rho_{B,i}$. It is now well-known that the $k$-symmetric extension of
$\rho_{AB}$ provides a hierarchy of separability criteria for
$\rho_{AB}$, which converges exactly to the set of separable states
when $k$ goes to infinity~\cite{doherty02a}. This result is
essentially given by the quantum de Finetti's
theorem~\cite{stormer1969symmetric,hudson1976locally,doherty02a,Renner2007,Christandl2007,harrow2013church}.
Our method can, in some sense, be thought of as a converse to the
$k$-symmetric extension criterion of separability. We will use
separability instead as a criterion to test $k$-symmetric
extendability of a bipartite state. This, however, does not cause any
circular reasoning problem---we can instead use other known
separability criteria, such as the positive partial transpose
condition~\cite{Per96,HHH96}, to give necessary conditions for the
$k$-symmetric extension problems.

In particular, our method computes a linear combination
$\tilde{\rho}_{AB}^{(k)}$ of the given density matrix $\rho_{AB}$ and
its reduced density matrix $\rho_A$. The separability of
$\tilde{\rho}_{AB}^{(k)}$ is then shown to be a necessary condition of
the corresponding $k$-symmetric extension problem for $\rho_{AB}$.

Interestingly, the condition can also be applied to the more general
setting of overlapping quantum marginal problems where the given
marginals on $A, B_i$ are different. We reduce them to the
$k$-symmetric extension problems of
$\frac{1}{k}\sum_{i=1}^k \rho_{AB_i}$. This averaging method may give
trivial conditions in adversarial situations. But it will nevertheless
provide non-trivial conditions better than many known results when the
given density matrices $\rho_{AB_i}$, though different, are related in
some way.

{\em Necessary conditions for the $k$-symmetric extension
  problems.---\/} Let $\H_A$, $\H_B$ be two Hilbert spaces of
dimension $d_A$ and $d_B$, respectively. For a Hilbert space $\H$, let
$\Density(\H)$ be the set of density matrices on $\H$. For a bipartite
state $\rho_{AB} \in \Density(\H_A\otimes \H_B)$, we consider the
following overlapping quantum marginal problem: whether there exists a
state
$\rho_{AB_1B_2\cdots B_k} \in \Density \bigl(H_A \otimes
(\bigotimes_{i=1}^k\H_{B_i}) \bigr)$
whose marginals on $A, B_i$ equal to $\rho_{AB}$ for all
$i=1, 2, \ldots, k$. The problem is also called the $k$-symmetric
extension problem of
$\rho_{AB}$~\cite{doherty02a,DPS04,DPS05,NOP09,BC12} and the global
state $\rho_{AB_1B_2\cdots B_k}$ is called a $k$-symmetric extension
of $\rho_{AB}$. If such a global state $\rho_{AB_1B_2\cdots B_k}$
exists, one can choose it to be invariant under permutations of
$B_1, B_2, \ldots, B_k$~\cite{doherty02a,myhr2011symmetric}.

If the state $\rho_{AB}$ is separable, then it is also obviously
$k$-symmetric extendable for any $k$. Interestingly, the converse of
the statement is also true. That is, if $\rho_{AB}$ is $k$-symmetric
extendable for all $k$, then $\rho_{AB}$ must be
separable~\cite{DPS04}. This provides a complete hierarchy of
separability criteria. The $k$-symmetric extension problem can be
formulated as a semidefinite programming (SDP), providing a numerical
procedure to detect entanglement in a mixed state (see
e.g.~\cite{VB96}).

In this paper, we want to know for a given $k$, whether $\rho_{AB}$ is
$k$-symmetric extendable. One can of course use the semidefinite
programming to solve the problem, but the size of the SDP formulation
will grow exponentially with $k$, rendering the approach impractical
even numerically for large $k$. We will instead use the separability
of some derived state $\tilde{\rho}_{AB}^{(k)}$ to detect the
$k$-extendability of $\rho_{AB}$. The important thing is that the
dimension of the state $\tilde{\rho}_{AB}^{(k)}$ is independent of
$k$.

For convenience, we will also consider a variant of the $k$-symmetric
extension problem called the $k$-bosonic extension problem. For
Hilbert spaces $\H_i$ of dimension $d$, let $\bigvee_{i=1}^k \H_i$ be
the symmetric subspace of $\bigotimes_{i=1}^k \H_i$. A state
$\rho_{AB}$ has a $k$-bosonic extension if it has a $k$-symmetric
extension $\rho_{AB_1B_2\cdots B_k}$ whose support on
$B_1, B_2, \ldots, B_k$ is in the symmetric subspace
$\bigvee_{i=1}^k \H_{B_i}$.

Our main observation is the following theorem. In the theorem,
$\H_A$ and $\H_B$ are two Hilbert spaces of dimension $d_A$ and $d_B$
respectively.

\begin{theorem}
  \label{thm:main}
  If a bipartite state $\rho_{AB} \in \Density(\H_A \otimes \H_B)$ has
  a $k$-symmetric extension, then the bipartite state
  \begin{equation}
    \label{eq:main}
    \tilde{\rho}^{(k)}_{AB} = \frac{1}{d_B^2+k} (d_B\rho_A\otimes
    I_B+k\rho_{AB})
  \end{equation}
  is separable.
\end{theorem}

In order to prove this theorem, we first recall the following lemma
~\cite{navascues2009power,doherty2014entanglement}.

\begin{lemma}
  \label{lm:boson}
  If a bipartite state $\rho_{AB} \in \Density(\H_A \otimes \H_B)$ has
  a $k$-bosonic extension, then the bipartite state
  \begin{equation}
    \label{eq:bosext}
    \hat{\rho}^{(k)}_{AB} = \frac{1}{d_B+k} (\rho_A\otimes I_B +
    k\rho_{AB})
  \end{equation}
  is separable.
\end{lemma}

We include a proof of Lemma~\ref{lm:boson} for completeness, which 
will directly lead to a proof of Theorem~\ref{thm:main} and a generalization
to the multi-party marginals case as discussed later.

\begin{proof}
  Let $\H_{B_i}$ be Hilbert spaces of dimension $d_B$ and let
  $\rho \in \Density(\bigvee_{i=1}^k \H_{B_i})$ be a state supported
  on the symmetric subspace $\bigvee_{i=1}^k \H_{B_i}$. Consider the
  following superoperator $\mathcal{E}$:
  \begin{equation}
    \label{eq:k1}
    \begin{split}
      \mathcal{E} (\rho) & = \int \bra{u}^{\otimes k}
      \rho\; \ket{u}^{\otimes k} \op{u}{u} \dif\mu(u),\\
      & = \Tr_{B_1\cdots B_k} \Bigl[ \bigl( I_B \otimes \rho \bigr)
      \int
      \op{u}{u}^{\otimes k+1} \dif\mu(u) \Bigr]\\
      & \propto \Tr_{B_1\cdots B_k} \Bigl[ \bigl( I_B \otimes \rho
      \bigr) \sum_{\pi\in S_{k+1}} W_\pi \Bigr],
    \end{split}
  \end{equation}
  where $\dif\mu(u)$ is the Haar measure over the pure states of
  $\H_B$ and $W_\pi$ is the permutation operator defined by
  \begin{equation*}
    W_\pi \ket{i_1, i_2, \ldots, i_k} = \ket{i_{\pi^{-1}(1)},
      i_{\pi^{-1}(2)}, \ldots, i_{\pi^{-1}(k)}}.
  \end{equation*}
  We claim that
  \begin{equation}
    \label{eq:claim}
    \mathcal{E} (\rho) \propto \tr(\rho) I_B + k\rho_B,
  \end{equation}
  for all state $\rho \in \Density(\bigvee_{i=1}^k \H_{B_i})$ where
  $\rho_B$ is the $1$-particle marginal of $\rho$. The claim follows
  from the Chiribella's theorem~\cite{chiribella2010quantum}; we give
  a proof here for its importance to our work. By the fact that any
  state $\rho$ supported on the symmetric subspace $\bigvee^k \H_B$
  can be written as the linear combination of states of the form
  $\op{\phi}{\phi}^{\otimes k}$ (see the Appendix
  of~\cite{chiribella2010quantum}), it suffices to prove the claim in
  Eq.~\eqref{eq:claim} for $\rho = \op{\phi}{\phi}^{\otimes k}$. For
  all $\pi \in S_k$,
  \begin{equation*}
    \Tr_{B_1\cdots B_k} \Bigl[ \bigl(I_B \otimes
    \op{\phi}{\phi}^{\otimes k} \bigr) W_\pi \Bigr] =
    \begin{cases}
      I_B & \text{if } \pi(1) = 1,\\
      \op{\phi}{\phi} & \text{otherwise}.
    \end{cases}
  \end{equation*}
  There are $k!$ permutations $\pi$ such that $\pi(1)=1$ and
  $k\cdot k!$ permutations $\pi(1)\ne 1$ and the claim follows from
  Eq.~\eqref{eq:k1}.

  If $\rho_{AB}$ has a $k$-bosonic extension
  $\rho_{AB_1B_2\cdots B_k}$, by Eq.~\eqref{eq:claim},
  \begin{equation*}
    \mathcal{I}_A\otimes \mathcal{E} (\rho_{AB_1B_2\cdots B_k})
    \propto \rho_A \otimes I_B + k\rho_{AB}.
  \end{equation*}
  The separability of $\hat{\rho}^{(k)}_{AB}$ then follows from the
  positive semidefinite property of $\rho_{AB_1B_2\cdots B_k}$ and
  Eq.~\eqref{eq:k1}.
\end{proof}

We now prove Theorem~\ref{thm:main}.
\begin{proof}[Proof of Theorem~\ref{thm:main}]
  Let
  $\rho \in \Density \bigl( \H_A\otimes (\bigotimes_{i=1}^k \H_{B_i})
  \bigr)$
  be the $k$-symmetric extension of $\rho_{AB}$. There exists a
  purification
  \begin{equation*}
    \ket{\Phi} \in \H_A \otimes \H_{A'} \otimes \Bigl[ \bigvee_{i=1}^k
    (\H_{B_i} \otimes \H_{B_i'}) \Bigr]
  \end{equation*}
  of $\rho$ where $d_{A'} = d_A$ and
  $d_{B_i'} = d_B$~\cite{Watrous2011}. State
  $\sigma = \op{\Phi}{\Phi}$ is the $k$-bosonic extension of its
  reduced density matrix $\sigma_{AA'BB'}$ on $A, A', B_1, B_1'$. By
  Lemma~\ref{lm:boson},
  \begin{equation*}
    \hat{\sigma}_{AA'BB'} = \frac{1}{d_B^2 + k} \bigl( \sigma_{AA'}
    \otimes I_{BB'} + k \sigma_{AA'BB'} \bigr)
  \end{equation*}
  is separable between $AA'$ and $BB'$. Tracing out the systems $A'$
  and $B'$, it follows that
  \begin{equation*}
    \tilde{\rho}^{(k)}_{AB} = \frac{1}{d_B^2 + k} \bigl( d_B \rho_A
    \otimes I_B + k \rho_{AB} \bigr)
  \end{equation*}
  is separable.
\end{proof}

{\em Examples of Bell-diagonal states.---\/} First consider the simple
case of $k=2$, and $A, B$ are qubit systems ($d_A=d_B=2$). Since for
any two-qubit state, the existence of a $2$-symmetric extension
implies that of a $2$-bosonic extension (see Proposition 21
of~\cite{myhr2011symmetric}), we can use the stronger condition of
Eq.~\eqref{eq:bosext} also for the symmetric extension problem. For
simplicity, we will investigate our condition for $2$-symmetric
extension for the class of Bell-diagonal states. A state $\rho_{AB}$
is Bell-diagonal if it is of the form
\begin{equation}
  \label{eq:bell-diag}
  \rho_{AB} = \sum_{i=1}^4 p_i \op{\Phi_i}{\Phi_i},
\end{equation}
where $p_i\in [0,1]$, $\sum_i p_i = 1$ and
\begin{align*}
  \ket{\Phi_1} & = (\ket{00}+\ket{11})/\sqrt{2},\; \ket{\Phi_2} =
                 (\ket{00}-\ket{11})/\sqrt{2},\\
  \ket{\Phi_3} & = (\ket{01}+\ket{10})/\sqrt{2},\; \ket{\Phi_4} =
                 (\ket{01}-\ket{10})/\sqrt{2}
\end{align*}
are the four Bell states.

A simple computation tells that our condition that
$\hat{\rho}^{(2)}_{AB}$ being separable is equivalent to
$p_i\in [0,3/4]$ for all $i=1,2,3,4$. This is a close approximation of
the exact condition of $2$-symmetric extendability given
in~\cite{ML09,chen2014symmetric,Cerf2000,niu1998optimal,cubitt2008structure}:
\begin{equation*}
  \frac{1}{2} \ge \sum_{i=1}^4 p_i^2 - 4\Bigl( \prod_{i=1}^4 p_i
  \Bigr)^{1/2}.
\end{equation*}
The regions of $p_1, p_2, p_3$ given by these two conditions are
plotted in Fig.~\ref{fig:ext-gold}. The volume of the exact set is
approximately $0.15115$ and the volume of the polytope given by our
condition is $0.15625$, which is only about $3\%$ larger.

\begin{figure}[htbp]
  \centering
  \begin{subfigure}[t]{\linewidth}
    \includegraphics[width=7cm]{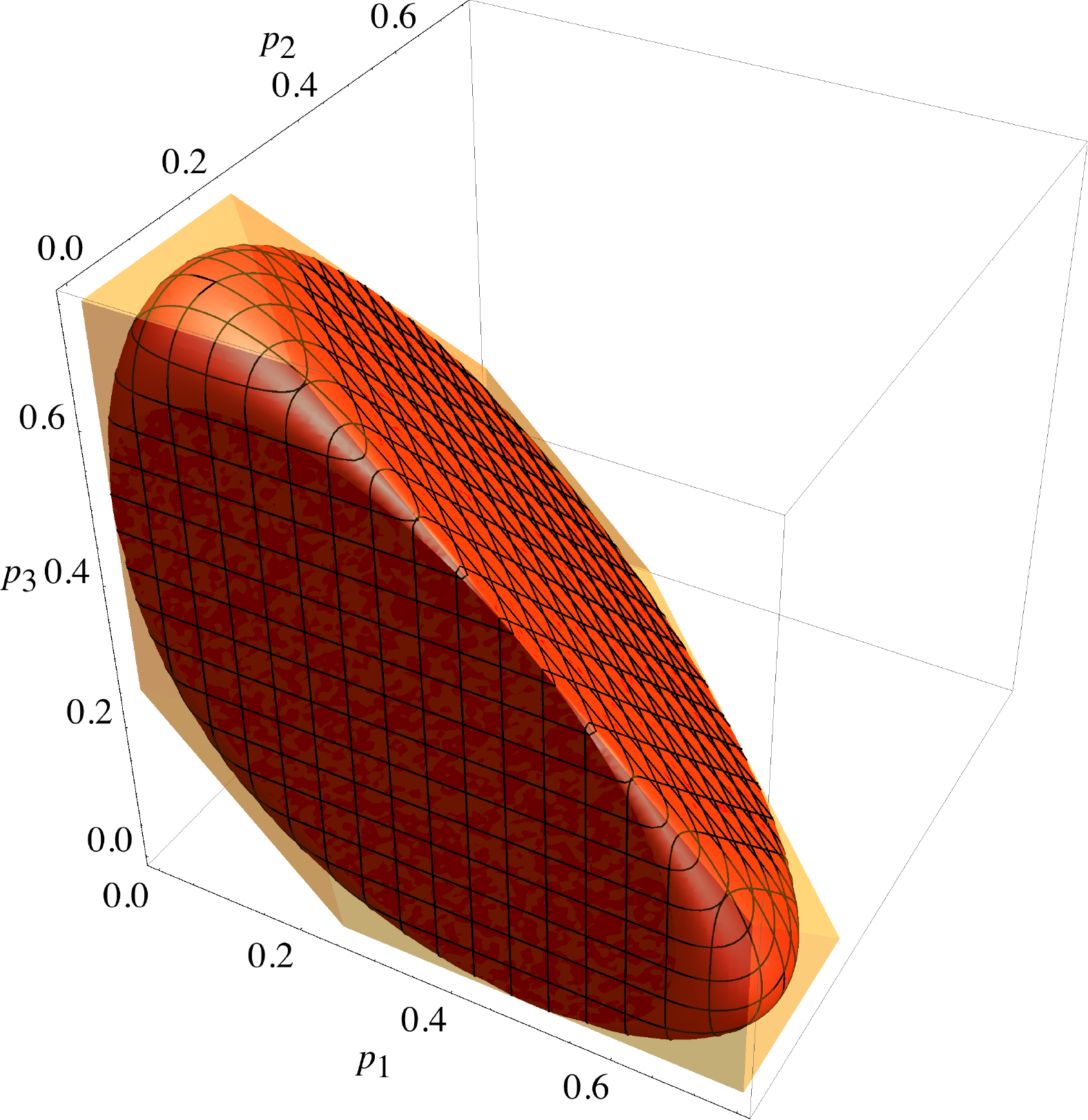}
    \caption{}
  \end{subfigure}
  \begin{subfigure}[t]{\linewidth}
    \includegraphics[width=6cm]{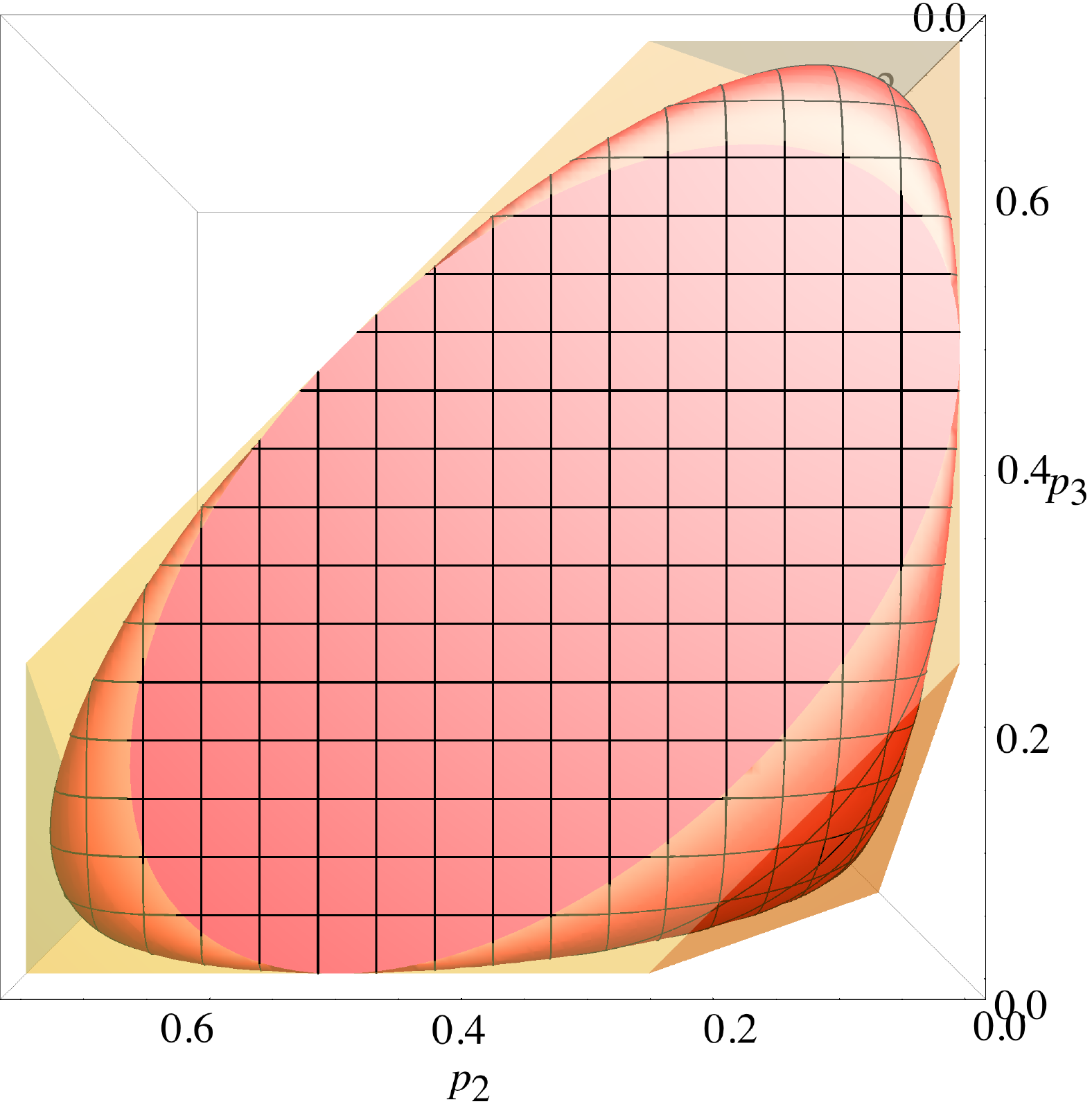}
    \caption{}
  \end{subfigure}
  \caption{(a) The polytope of yellow color characterized by
    $0\le p_i \le 3/4$ and $1/4 \le p_1 + p_2 + p_3 \le 1$ is the
    condition given by the separability of $\hat{\rho}_{AB}^{(2)}$.
    The convex set of red color is given by the necessary and
    sufficient condition for $2$-symmetric extension of Bell-diagonal
    states. (b) is the left view of the
    same figure.}
  \label{fig:ext-gold}
\end{figure}

For comparison purposes, we have also plotted the conditions given by
the strong subadditivity (SSA). For Bell-diagonal states, the SSA
condition simplifies to $S(AB) \ge 1$. We find that our condition and
the SSA condition are incomparable---the non-extendability can
sometimes be detected by our condition but not the SSA condition, and
vice versa. See Fig.~\ref{fig:ext-ssa} for details.

\begin{figure}[htbp]
  \centering
  \includegraphics[width=6cm]{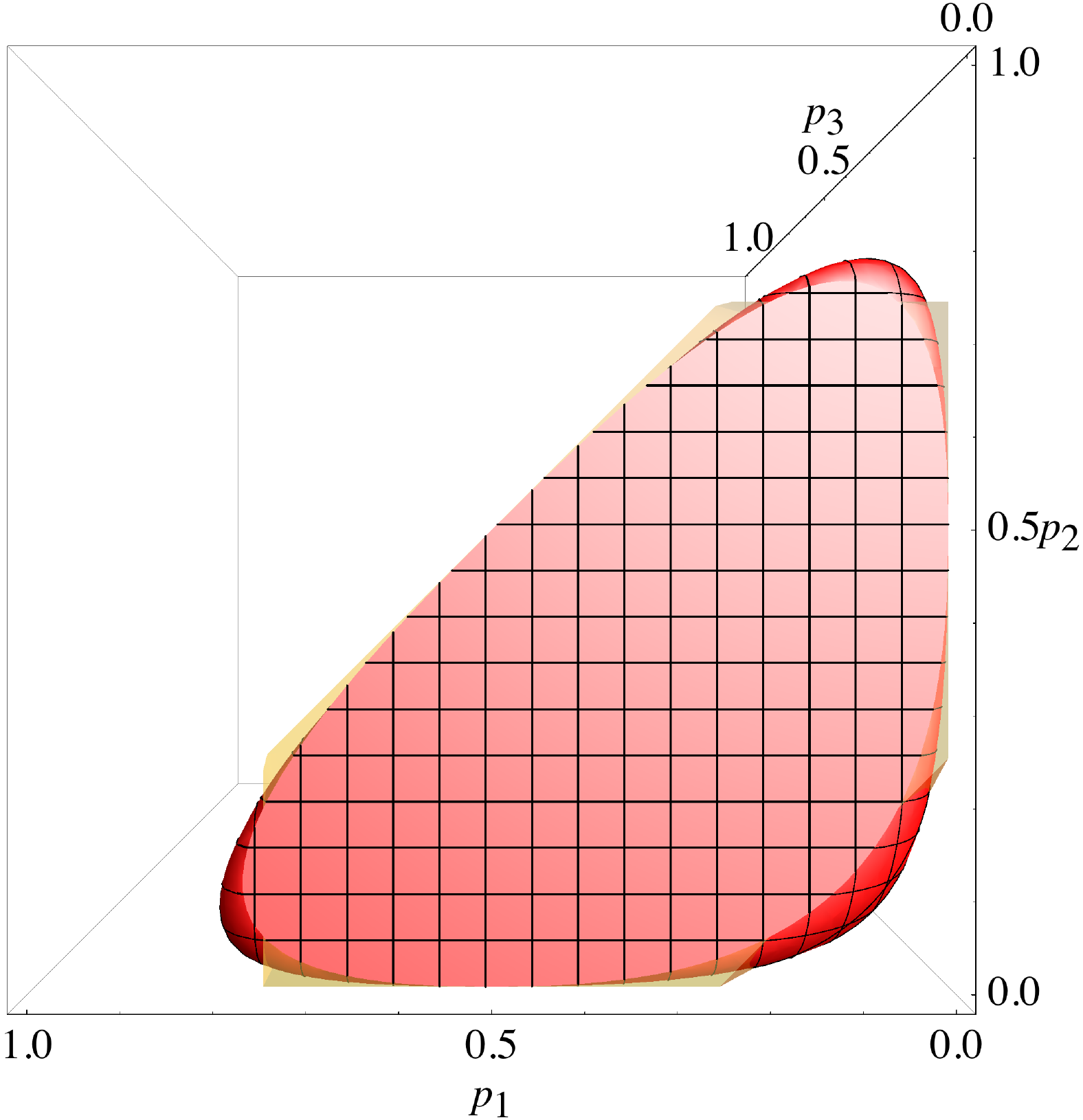}
  \caption{The two convex sets of $(p_1, p_2, p_3)$ corresponding to
    the condition given by the separability condition of
    $\rho^{(2)}_{AB}$ (the polytope of yellow color) and the SSA
    condition (the convex set of red color).}
  \label{fig:ext-ssa}
\end{figure}

{\em Examples of Werner states.---\/} In our next example, we analyze
our conditions for the $k$-symmetric extension problem of the Werner
states~\cite{Wer89,werner1990remarks}. A two-qudit Werner state is a state invariant
under the $U\otimes U$ operator for all unitary $U\in \Unitary(d)$ and
has the following form
\begin{equation*}
  \rho_W(\psi^{-}) = \frac{1+\psi^{-}}{2} \rho^{+} +
  \frac{1-\psi^{-}}{2} \rho^{-},
\end{equation*}
where $\psi^{-} \in [-1,1]$ is the parameter, $\rho^{+}$ and
$\rho^{-}$ are the states proportional to the projection of the
symmetric subspace $\vee^2 \complex^d$ and anti-symmetric subspace
$\wedge^2 \complex^d$ respectively. The Werner state
$\rho_W(\psi^{-})$ is separable if and only if $\psi^{-} \ge 0$. The
state $\tilde{\rho}_{W}^{(k)}(\psi^{-})$ is separable when
$\psi^{-} \ge -d/k$. Therefore, by Theorem~\ref{thm:main},
$\rho_W(\psi^{-})$ is not $k$-symmetric extendable if
$\psi^{-} < -d/k$. We note that our bound, though not optimal, is a
close approximation of the necessary and sufficient condition
$\psi^{-} \ge -(d-1)/k$ proved in~\cite{JV13} for the $k$-symmetric
extendability of Werner states. This also proves that the
$k$-symmetric extension and $k$-bosonic extension problems are
generally different. In particular, it also implies that the $d_B$ in
the linear combination in Eq.~\eqref{eq:main} is essential for the
$k$-symmetric extension problem.

{\em Applications to the overlapping marginal problems.---\/} We now
extend our method to the more general situation with different
marginals on $A, B_i$. That is, one asks whether there exists a state
$\rho_{AB_1B_2\cdots B_k} \in \Density \bigl(\H_A \otimes
(\bigotimes_{i=1}^k \H_{B_i}) \bigr)$
whose marginals on $A, B_i$ is the given density matrices
$\rho_{AB_i}$ for all $i=1,2,\ldots, k$. This consistency problem for
bipartite marginals is of vital importance in many-body physics and
quantum chemistry, where the Hamiltonians of the system in general
involve only two-body interactions~\cite{Col63,Erd72,ZCZW15}.

In order to use the necessary condition derived in the previous
section, we observe the following fact.
\begin{lemma}
  \label{lem:consistency}
  If the marginals $\rho_{AB_i}$ with $i=1,2,\ldots, k$ are
  consistent, then the bipartite state
  \begin{equation}
    \rho_{AB} = \frac{1}{k} \sum_{i=1}^{k} \rho_{AB_i}
  \end{equation}
  has $k$-symmetric extension.
\end{lemma}

\begin{proof}
  If $\rho_{AB_i}$ with $i=1,2,\ldots, k$ are consistent, then there
  exists a state
  $\rho_{AB_1B_2\cdots B_k} \in \Density(\H_A\otimes\H_B^{\otimes
    k})$,
  such that its reduced density matrix on the system $AB_i$ is
  $\rho_{AB_i}$ for all $i=1,2,\ldots, k$. Now consider the state
  \begin{equation}
    \rho'_{AB_1B_2\cdots B_k} = \frac{1}{k!} \sum_{\pi\in{S}_k}
    \rho_{AB_{\pi(1)} B_{\pi(2)} \cdots B_{\pi(k)}},
  \end{equation}
  where $S_k$ is the symmetric group of $k$ elements. Then
  $\rho'_{AB_1B_2\cdots B_k}$ is a $k$-symmetric extension of
  $\rho_{AB}$.
\end{proof}

This then allows us to use Theorem~\ref{thm:main} and Lemma~\ref{lm:boson}
to detect consistency of bipartite marginals. Consider the example of
a three-qubit system with $\rho_{AB} = \rho_W(\psi^{-}_1)$, and
$\rho_{AC} = \rho_W(\psi^{-}_2)$ for $\psi^{-}_i \in [-1,1]$ , both of
which are two-qubit Werner states. For two-qubit states, $2$-symmetric
extendability implies $2$-bosonic extendability. Hence, we can use the
condition of Eq.~\eqref{eq:bosext}, which implies that $\rho_{AB}$ and
$\rho_{AC}$ are consistent only if
$(\psi^{-}_1+\psi^{-}_2)/2 \ge -1/2$. This in fact gives a
quantitative entanglement monogamy
inequality~\cite{Werner1989,Coffman2000,Terhal2003,Bae2006,Osborne2006}
for Werner states.

\begin{figure}[htbp]
  \centering
  \begin{subfigure}[t]{.45\linewidth}
    \includegraphics[width=4cm]{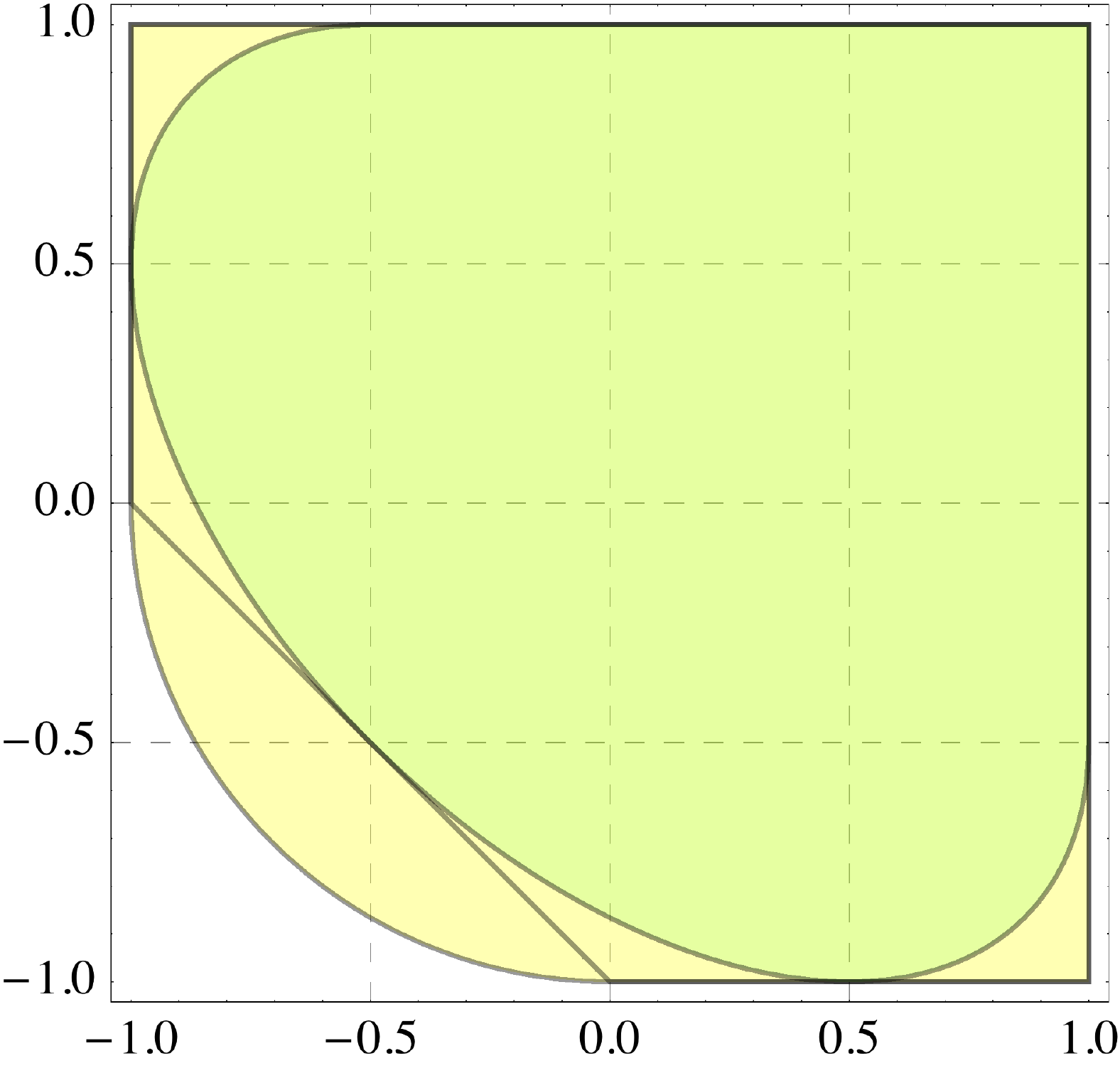}
    \caption{}
    \label{fig:wer-ckw}
  \end{subfigure}
  \begin{subfigure}[t]{.45\linewidth}
    \includegraphics[width=4cm]{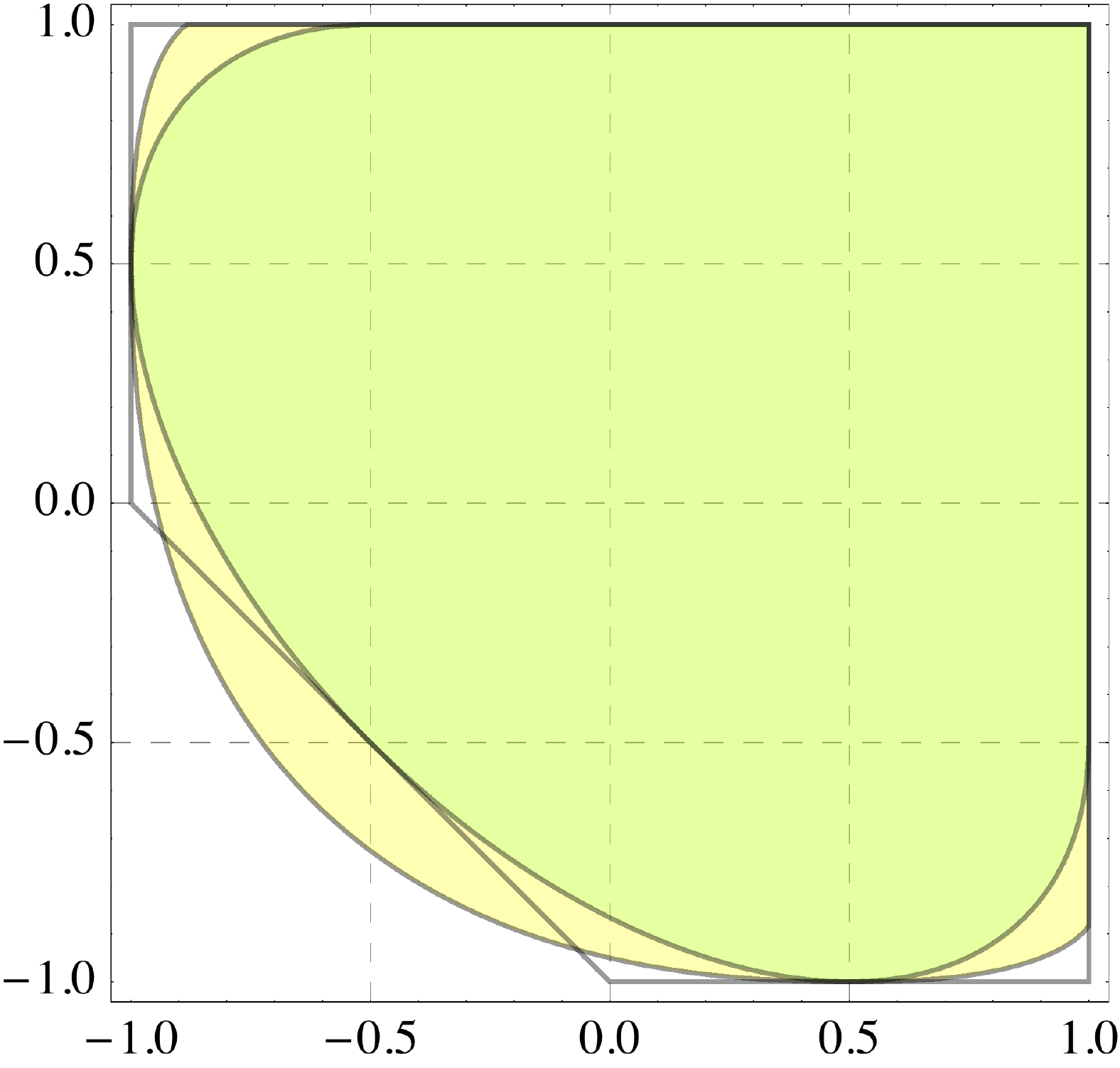}
    \caption{}
    \label{fig:wer-ssa}
  \end{subfigure}
  \caption{The green region is the exact condition for two Werner
    states to be consistent. The pentagon defined by
    $\psi^{-}_1 + \psi^{-}_2 \ge -1$ and $-1\le \psi^{-}_i \le 1$ is
    the condition given by our criterion. (a) is the condition given
    by the CKW entanglement monogamy inequality, and (b) is the SSA
    condition.}
  \label{fig:werner}
\end{figure}

We compare our condition to that given by the Coffman-Kundu-Wootters
(CKW) entanglement monogamy inequality~\cite{Coffman2000},
\begin{equation*}
  C_{AB}^2 + C_{AC}^2 \leq C_{A(BC)}^2,
\end{equation*}
where $C_{AB} = \max \{0, -\psi^{-}_1\}$,
$C_{AC} = \max \{0,-\psi^{-}_2\}$ are the
concurrences~\cite{HW97,Woo98} between $A, B$ and $A, C$ respectively,
while $C_{A(BC)} = 1$ is the concurrence between subsystems $A$ and
$BC$ for Werner states. As shown in Fig.~\ref{fig:wer-ckw}, our
condition (the pentagon defined by $\psi^{-}_1 + \psi^{-}_2 \ge -1$
and $-1\le \psi^{-}_i \le 1$) is always better than the condition
given by the CKW inequality (the union of the yellow and green
regions).

We have also computed the SSA condition for this particular case and
plotted the regions of the SSA condition and our condition in
Fig.~\ref{fig:wer-ssa}. Again, the SSA condition (the union of the
yellow and green regions) is incomparable with ours.

{\em Generalizations.---\/} Our method extends to the following more
general settings. Let
$\rho_{AB_1B_2\cdots B_r} \in \Density(\H_A \otimes \H_B^{\otimes r})$
be a given density matrix. The $(r,k)$-bosonic extension problem of
$\rho_{AB_1B_2\cdots B_r}$ asks whether there is a global state
$\rho_{AB_1B_2\cdots B_k} \in \Density \bigl( \H_A \otimes
(\bigvee^k\H_B) \bigr)$
whose marginal on $A, B_1, B_2, \ldots, B_r$ is
$\rho_{AB_1B_2\cdots B_r}$. Following a similar argument as in the
proof of Lemma~\ref{lm:boson} and using the Chiribella's
theorem~\cite{chiribella2010quantum,harrow2013church}, one obtains a
necessary condition generalizing Lemma~\ref{lm:boson}. Namely,
\begin{equation}
  \label{eq:rext}
  \hat{\rho}_{AB_1B_2\cdots B_r}^{(k)} = \sum_{s=0}^{r}p_s(k,d_B,r)
  \mathcal{I}_A\otimes \mathcal{E}_s(\rho_{AB_1\cdots B_s})
\end{equation}
is an $r+1$-party separable state. Here,
\begin{equation}
  p_s(k,d,r) = \frac{{k\choose s}{d+r-1 \choose r-s}}{{d+k+r-1 \choose
      r}},
\end{equation}
is a distribution satisfying $\sum_{s=0}^r p_s =1$, and
$\mathcal{E}_s$ is the superoperator given by
\begin{equation}
  \mathcal{E}_s(\rho)=\frac{d_s}{d_r}\Pi^{+}_{r}(\rho_s\otimes
  I^{\otimes (r-s)})\Pi^{+}_{r},
\end{equation}
where $ d_r={d+r-1 \choose r}, $ and $\Pi^{+}_{r}$ is the projection
onto the symmetric subspace $\vee^r \complex^d$.

At the moment, however, we do not know how to generalize the formula
in Theorem~\ref{thm:main} to this multi-party setting as the procedure
of tracing out $A',B_1',\ldots, B_r'$ does not commute with the
projection $\Pi^{+}_r$ in general. We leave it as an open problem for
future work.

{\em Summary and discussion.---\/} We have proposed a method to detect
consistency of overlapping quantum marginals. The key idea is to
construct some other density matrix from the linear combinations of
the local density matrices and test the separability of the derived
density matrix. Our idea is closely related to the finite quantum de
Finetti's
theorem~\cite{diaconis1980finite,doherty02a,Renner2007,Christandl2007},
which states that the $r$-particle marginal of a symmetric
$N$-particle state cannot be too far from an $r$-particle separate
state, with a distance bounded by $O(1/N)$ for fixed $d$ and $r$.
Therefore, if an $r$-particle state is too far from a separable state,
then it cannot be the marginal of a symmetric $N$-particle state.
However, to directly check the distance to the nearest separable state
is not easy. Moreover, the bound given in the known versions of finite
quantum de Finetti's theorem are in general not tight, so when $N$ is
small those bound may not be useful.

For comparison, our method gives simple necessarily conditions, which
are evidently good even for $N$ small. Our method can also lead to
improved bound in the finite de Finetti's theorem. For instance, as a
direct consequence of Theorem~\ref{thm:main}, 
we can obtain that for any $k$-symmetric
extendible state $\rho_{AB}$, its distance to separable states is
upper bounded by
\begin{equation}
  \min_{\rho\in \Sep}\norm{\rho_{AB} - \rho}_1 \le \norm{ \rho_{AB} -
    \tilde{\rho}^{(k)}_{AB}}_1 \le \frac{2d_B^2}{d_B^2+k},
\end{equation}
which slightly improves that of~\cite{Christandl2007}.

Another direct application is that in Lemma~\ref{lm:boson} if we
choose $k=1$, then from Eq.~\eqref{eq:bosext}, we get that for any
bipartite state $\rho_{AB}$, the state
\begin{equation}
\label{eq:sep}
  \sigma_{AB} = \frac{1}{d_B+1}(\rho_A\otimes I_B+\rho_{AB}),
\end{equation}
is always separable. Notice that Eq.~\eqref{eq:sep} implies that
$\sigma_A=\rho_A$, so we have 
$
({d_B+1}) \sigma_{AB} - \sigma_A\otimes I_B=\rho_{AB}\geq 0.
$
This gives an interesting sufficient condition
of separability for $\sigma_{AB}$: if 
$
({d_B+1})\sigma_{AB}\geq \sigma_A\otimes I_B,
$
then $\sigma_{AB}$ is separable. We may also compare this with the known
necessary condition of separability for $\sigma_{AB}$~\cite{HH99,CAG99}: if $\sigma_{AB}$ is separable, then 
$
\sigma_{AB}\leq \sigma_A\otimes I_B.
$
 
\section*{Acknowledgment}

We thank John Watrous for helpful discussions. ZJ acknowledges support
from NSERC and ARO. NY and BZ are supported by NSERC.

\bibliography{marginal}

\begin{thebibliography}{50}%
\makeatletter
\providecommand \@ifxundefined [1]{%
 \@ifx{#1\undefined}
}%
\providecommand \@ifnum [1]{%
 \ifnum #1\expandafter \@firstoftwo
 \else \expandafter \@secondoftwo
 \fi
}%
\providecommand \@ifx [1]{%
 \ifx #1\expandafter \@firstoftwo
 \else \expandafter \@secondoftwo
 \fi
}%
\providecommand \natexlab [1]{#1}%
\providecommand \enquote  [1]{``#1''}%
\providecommand \bibnamefont  [1]{#1}%
\providecommand \bibfnamefont [1]{#1}%
\providecommand \citenamefont [1]{#1}%
\providecommand \href@noop [0]{\@secondoftwo}%
\providecommand \href [0]{\begingroup \@sanitize@url \@href}%
\providecommand \@href[1]{\@@startlink{#1}\@@href}%
\providecommand \@@href[1]{\endgroup#1\@@endlink}%
\providecommand \@sanitize@url [0]{\catcode `\\12\catcode `\$12\catcode
  `\&12\catcode `\#12\catcode `\^12\catcode `\_12\catcode `\%12\relax}%
\providecommand \@@startlink[1]{}%
\providecommand \@@endlink[0]{}%
\providecommand \url  [0]{\begingroup\@sanitize@url \@url }%
\providecommand \@url [1]{\endgroup\@href {#1}{\urlprefix }}%
\providecommand \urlprefix  [0]{URL }%
\providecommand \Eprint [0]{\href }%
\providecommand \doibase [0]{http://dx.doi.org/}%
\providecommand \selectlanguage [0]{\@gobble}%
\providecommand \bibinfo  [0]{\@secondoftwo}%
\providecommand \bibfield  [0]{\@secondoftwo}%
\providecommand \translation [1]{[#1]}%
\providecommand \BibitemOpen [0]{}%
\providecommand \bibitemStop [0]{}%
\providecommand \bibitemNoStop [0]{.\EOS\space}%
\providecommand \EOS [0]{\spacefactor3000\relax}%
\providecommand \BibitemShut  [1]{\csname bibitem#1\endcsname}%
\let\auto@bib@innerbib\@empty
%</preamble>
\bibitem [{\citenamefont {Klyachko}(2006)}]{Kly06}%
  \BibitemOpen
  \bibfield  {author} {\bibinfo {author} {\bibfnamefont {A.~A.}\ \bibnamefont
  {Klyachko}},\ }\href {\doibase 10.1088/1742-6596/36/1/014} {\bibfield
  {journal} {\bibinfo  {journal} {J. Phys.: Conf. Ser.}\ }\textbf {\bibinfo
  {volume} {36}},\ \bibinfo {pages} {72} (\bibinfo {year} {2006})}\BibitemShut
  {NoStop}%
\bibitem [{\citenamefont {Coleman}(1963)}]{Col63}%
  \BibitemOpen
  \bibfield  {author} {\bibinfo {author} {\bibfnamefont {A.~J.}\ \bibnamefont
  {Coleman}},\ }\href {\doibase 10.1103/RevModPhys.35.668} {\bibfield
  {journal} {\bibinfo  {journal} {Rev. Mod. Phys.}\ }\textbf {\bibinfo {volume}
  {35}},\ \bibinfo {pages} {668} (\bibinfo {year} {1963})}\BibitemShut
  {NoStop}%
\bibitem [{\citenamefont {{Erdahl}}(1972)}]{Erd72}%
  \BibitemOpen
  \bibfield  {author} {\bibinfo {author} {\bibfnamefont {R.~M.}\ \bibnamefont
  {{Erdahl}}},\ }\href {\doibase 10.1063/1.1665885} {\bibfield  {journal}
  {\bibinfo  {journal} {J. Math. Phys.}\ }\textbf {\bibinfo {volume} {13}},\
  \bibinfo {pages} {1608} (\bibinfo {year} {1972})}\BibitemShut {NoStop}%
\bibitem [{\citenamefont {Liu}(2006)}]{Liu06}%
  \BibitemOpen
  \bibfield  {author} {\bibinfo {author} {\bibfnamefont {Y.-K.}\ \bibnamefont
  {Liu}},\ }in\ \href {\doibase 10.1007/11830924_40} {\emph {\bibinfo
  {booktitle} {Approximation, Randomization, and Combinatorial Optimization.
  Algorithms and Techniques}}},\ \bibinfo {series} {Lecture Notes in Computer
  Science}, Vol.\ \bibinfo {volume} {4110},\ \bibinfo {editor} {edited by\
  \bibinfo {editor} {\bibfnamefont {J.}~\bibnamefont {Diaz}}, \bibinfo {editor}
  {\bibfnamefont {K.}~\bibnamefont {Jansen}}, \bibinfo {editor} {\bibfnamefont
  {J.~D.}\ \bibnamefont {Rolim}}, \ and\ \bibinfo {editor} {\bibfnamefont
  {U.}~\bibnamefont {Zwick}}}\ (\bibinfo  {publisher} {Springer Berlin
  Heidelberg},\ \bibinfo {year} {2006})\ pp.\ \bibinfo {pages}
  {438--449}\BibitemShut {NoStop}%
\bibitem [{\citenamefont {Liu}\ \emph {et~al.}(2007)\citenamefont {Liu},
  \citenamefont {Christandl},\ and\ \citenamefont {Verstraete}}]{LCV07}%
  \BibitemOpen
  \bibfield  {author} {\bibinfo {author} {\bibfnamefont {Y.-K.}\ \bibnamefont
  {Liu}}, \bibinfo {author} {\bibfnamefont {M.}~\bibnamefont {Christandl}}, \
  and\ \bibinfo {author} {\bibfnamefont {F.}~\bibnamefont {Verstraete}},\
  }\href {\doibase 10.1103/PhysRevLett.98.110503} {\bibfield  {journal}
  {\bibinfo  {journal} {Phys. Rev. Lett.}\ }\textbf {\bibinfo {volume} {98}},\
  \bibinfo {pages} {110503} (\bibinfo {year} {2007})}\BibitemShut {NoStop}%
\bibitem [{\citenamefont {Wei}\ \emph {et~al.}(2010)\citenamefont {Wei},
  \citenamefont {Mosca},\ and\ \citenamefont {Nayak}}]{WMN10}%
  \BibitemOpen
  \bibfield  {author} {\bibinfo {author} {\bibfnamefont {T.-C.}\ \bibnamefont
  {Wei}}, \bibinfo {author} {\bibfnamefont {M.}~\bibnamefont {Mosca}}, \ and\
  \bibinfo {author} {\bibfnamefont {A.}~\bibnamefont {Nayak}},\ }\href
  {\doibase 10.1103/PhysRevLett.104.040501} {\bibfield  {journal} {\bibinfo
  {journal} {Phys. Rev. Lett.}\ }\textbf {\bibinfo {volume} {104}},\ \bibinfo
  {pages} {040501} (\bibinfo {year} {2010})}\BibitemShut {NoStop}%
\bibitem [{\citenamefont {Klyachko}(2004)}]{Kly04}%
  \BibitemOpen
  \bibfield  {author} {\bibinfo {author} {\bibfnamefont {A.}~\bibnamefont
  {Klyachko}},\ }\href@noop {} {\  (\bibinfo {year} {2004})}\BibitemShut
  {NoStop}%
\bibitem [{\citenamefont {{Altunbulak}}\ and\ \citenamefont
  {{Klyachko}}(2008)}]{AK08}%
  \BibitemOpen
  \bibfield  {author} {\bibinfo {author} {\bibfnamefont {M.}~\bibnamefont
  {{Altunbulak}}}\ and\ \bibinfo {author} {\bibfnamefont {A.}~\bibnamefont
  {{Klyachko}}},\ }\href {\doibase 10.1007/s00220-008-0552-z} {\bibfield
  {journal} {\bibinfo  {journal} {Comm. Math. Phys.}\ }\textbf {\bibinfo
  {volume} {282}},\ \bibinfo {pages} {287} (\bibinfo {year}
  {2008})}\BibitemShut {NoStop}%
\bibitem [{\citenamefont {Schilling}\ \emph {et~al.}(2013)\citenamefont
  {Schilling}, \citenamefont {Gross},\ and\ \citenamefont
  {Christandl}}]{SGC13}%
  \BibitemOpen
  \bibfield  {author} {\bibinfo {author} {\bibfnamefont {C.}~\bibnamefont
  {Schilling}}, \bibinfo {author} {\bibfnamefont {D.}~\bibnamefont {Gross}}, \
  and\ \bibinfo {author} {\bibfnamefont {M.}~\bibnamefont {Christandl}},\
  }\href {\doibase 10.1103/PhysRevLett.110.040404} {\bibfield  {journal}
  {\bibinfo  {journal} {Phys. Rev. Lett.}\ }\textbf {\bibinfo {volume} {110}},\
  \bibinfo {pages} {040404} (\bibinfo {year} {2013})}\BibitemShut {NoStop}%
\bibitem [{\citenamefont {Walter}\ \emph {et~al.}(2013)\citenamefont {Walter},
  \citenamefont {Doran}, \citenamefont {Gross},\ and\ \citenamefont
  {Christandl}}]{walter2013entanglement}%
  \BibitemOpen
  \bibfield  {author} {\bibinfo {author} {\bibfnamefont {M.}~\bibnamefont
  {Walter}}, \bibinfo {author} {\bibfnamefont {B.}~\bibnamefont {Doran}},
  \bibinfo {author} {\bibfnamefont {D.}~\bibnamefont {Gross}}, \ and\ \bibinfo
  {author} {\bibfnamefont {M.}~\bibnamefont {Christandl}},\ }\href {\doibase
  10.1126/science.1232957} {\bibfield  {journal} {\bibinfo  {journal}
  {Science}\ }\textbf {\bibinfo {volume} {340}},\ \bibinfo {pages} {1205}
  (\bibinfo {year} {2013})}\BibitemShut {NoStop}%
\bibitem [{\citenamefont {Sawicki}\ \emph {et~al.}(2014)\citenamefont
  {Sawicki}, \citenamefont {Oszmaniec},\ and\ \citenamefont
  {Ku{\'s}}}]{sawicki2014convexity}%
  \BibitemOpen
  \bibfield  {author} {\bibinfo {author} {\bibfnamefont {A.}~\bibnamefont
  {Sawicki}}, \bibinfo {author} {\bibfnamefont {M.}~\bibnamefont {Oszmaniec}},
  \ and\ \bibinfo {author} {\bibfnamefont {M.}~\bibnamefont {Ku{\'s}}},\ }\href
  {\doibase 10.1142/S0129055X14500044} {\bibfield  {journal} {\bibinfo
  {journal} {Rev. Mod. Phys.}\ }\textbf {\bibinfo {volume} {26}},\ \bibinfo
  {pages} {1450004} (\bibinfo {year} {2014})}\BibitemShut {NoStop}%
\bibitem [{\citenamefont {{Carlen}}\ \emph {et~al.}(2013)\citenamefont
  {{Carlen}}, \citenamefont {{Lebowitz}},\ and\ \citenamefont
  {{Lieb}}}]{CLL13}%
  \BibitemOpen
  \bibfield  {author} {\bibinfo {author} {\bibfnamefont {E.~A.}\ \bibnamefont
  {{Carlen}}}, \bibinfo {author} {\bibfnamefont {J.~L.}\ \bibnamefont
  {{Lebowitz}}}, \ and\ \bibinfo {author} {\bibfnamefont {E.~H.}\ \bibnamefont
  {{Lieb}}},\ }\href {\doibase 10.1063/1.4808218} {\bibfield  {journal}
  {\bibinfo  {journal} {J. Math. Phys.}\ }\textbf {\bibinfo {volume} {54}},\
  \bibinfo {pages} {062103} (\bibinfo {year} {2013})}\BibitemShut {NoStop}%
\bibitem [{\citenamefont {Coffman}\ \emph {et~al.}(2000)\citenamefont
  {Coffman}, \citenamefont {Kundu},\ and\ \citenamefont
  {Wootters}}]{Coffman2000}%
  \BibitemOpen
  \bibfield  {author} {\bibinfo {author} {\bibfnamefont {V.}~\bibnamefont
  {Coffman}}, \bibinfo {author} {\bibfnamefont {J.}~\bibnamefont {Kundu}}, \
  and\ \bibinfo {author} {\bibfnamefont {W.~K.}\ \bibnamefont {Wootters}},\
  }\href {\doibase 10.1103/PhysRevA.61.052306} {\bibfield  {journal} {\bibinfo
  {journal} {Phys. Rev. A}\ }\textbf {\bibinfo {volume} {61}},\ \bibinfo
  {pages} {052306} (\bibinfo {year} {2000})}\BibitemShut {NoStop}%
\bibitem [{\citenamefont {{Smith}}(1965)}]{Smi65}%
  \BibitemOpen
  \bibfield  {author} {\bibinfo {author} {\bibfnamefont {D.~W.}\ \bibnamefont
  {{Smith}}},\ }\href {\doibase 10.1063/1.1701504} {\bibfield  {journal}
  {\bibinfo  {journal} {J. Chem. Phys.}\ }\textbf {\bibinfo {volume} {43}},\
  \bibinfo {pages} {258} (\bibinfo {year} {1965})}\BibitemShut {NoStop}%
\bibitem [{\citenamefont {Chen}\ \emph {et~al.}(2014)\citenamefont {Chen},
  \citenamefont {Ji}, \citenamefont {Kribs}, \citenamefont {L{\"u}tkenhaus},\
  and\ \citenamefont {Zeng}}]{chen2014symmetric}%
  \BibitemOpen
  \bibfield  {author} {\bibinfo {author} {\bibfnamefont {J.}~\bibnamefont
  {Chen}}, \bibinfo {author} {\bibfnamefont {Z.}~\bibnamefont {Ji}}, \bibinfo
  {author} {\bibfnamefont {D.}~\bibnamefont {Kribs}}, \bibinfo {author}
  {\bibfnamefont {N.}~\bibnamefont {L{\"u}tkenhaus}}, \ and\ \bibinfo {author}
  {\bibfnamefont {B.}~\bibnamefont {Zeng}},\ }\href {\doibase
  10.1103/PhysRevA.90.032318} {\bibfield  {journal} {\bibinfo  {journal} {Phys.
  Rev. A}\ }\textbf {\bibinfo {volume} {90}},\ \bibinfo {pages} {032318}
  (\bibinfo {year} {2014})}\BibitemShut {NoStop}%
\bibitem [{\citenamefont {Doherty}\ \emph {et~al.}(2002)\citenamefont
  {Doherty}, \citenamefont {Parrilo},\ and\ \citenamefont
  {Spedalieri}}]{doherty02a}%
  \BibitemOpen
  \bibfield  {author} {\bibinfo {author} {\bibfnamefont {A.~C.}\ \bibnamefont
  {Doherty}}, \bibinfo {author} {\bibfnamefont {P.~A.}\ \bibnamefont
  {Parrilo}}, \ and\ \bibinfo {author} {\bibfnamefont {F.~M.}\ \bibnamefont
  {Spedalieri}},\ }\href {\doibase 10.1103/PhysRevLett.88.187904} {\bibfield
  {journal} {\bibinfo  {journal} {Phys. Rev. Lett.}\ }\textbf {\bibinfo
  {volume} {88}},\ \bibinfo {pages} {187904} (\bibinfo {year}
  {2002})}\BibitemShut {NoStop}%
\bibitem [{\citenamefont {Stormer}(1969)}]{stormer1969symmetric}%
  \BibitemOpen
  \bibfield  {author} {\bibinfo {author} {\bibfnamefont {E.}~\bibnamefont
  {Stormer}},\ }\href {\doibase 10.1016/0022-1236(69)90050-0} {\bibfield
  {journal} {\bibinfo  {journal} {J. Funct. Anal.}\ }\textbf {\bibinfo {volume}
  {3}},\ \bibinfo {pages} {48} (\bibinfo {year} {1969})}\BibitemShut {NoStop}%
\bibitem [{\citenamefont {Hudson}\ and\ \citenamefont
  {Moody}(1976)}]{hudson1976locally}%
  \BibitemOpen
  \bibfield  {author} {\bibinfo {author} {\bibfnamefont {R.}~\bibnamefont
  {Hudson}}\ and\ \bibinfo {author} {\bibfnamefont {G.}~\bibnamefont {Moody}},\
  }\href {\doibase 10.1007/BF00534784} {\bibfield  {journal} {\bibinfo
  {journal} {Probab. Theory Rel.}\ }\textbf {\bibinfo {volume} {33}},\ \bibinfo
  {pages} {343} (\bibinfo {year} {1976})}\BibitemShut {NoStop}%
\bibitem [{\citenamefont {Renner}(2007)}]{Renner2007}%
  \BibitemOpen
  \bibfield  {author} {\bibinfo {author} {\bibfnamefont {R.}~\bibnamefont
  {Renner}},\ }\href {\doibase 10.1038/nphys684} {\bibfield  {journal}
  {\bibinfo  {journal} {Nat. Phys.}\ }\textbf {\bibinfo {volume} {3}},\
  \bibinfo {pages} {645} (\bibinfo {year} {2007})}\BibitemShut {NoStop}%
\bibitem [{\citenamefont {Christandl}\ \emph {et~al.}(2007)\citenamefont
  {Christandl}, \citenamefont {Konig}, \citenamefont {Mitchison},\ and\
  \citenamefont {Renner}}]{Christandl2007}%
  \BibitemOpen
  \bibfield  {author} {\bibinfo {author} {\bibfnamefont {M.}~\bibnamefont
  {Christandl}}, \bibinfo {author} {\bibfnamefont {R.}~\bibnamefont {Konig}},
  \bibinfo {author} {\bibfnamefont {G.}~\bibnamefont {Mitchison}}, \ and\
  \bibinfo {author} {\bibfnamefont {R.}~\bibnamefont {Renner}},\ }\href
  {\doibase 10.1007/s00220-007-0189-3} {\bibfield  {journal} {\bibinfo
  {journal} {Comm. Math. Phys}\ }\textbf {\bibinfo {volume} {273}},\ \bibinfo
  {pages} {473} (\bibinfo {year} {2007})}\BibitemShut {NoStop}%
\bibitem [{\citenamefont {Harrow}(2013)}]{harrow2013church}%
  \BibitemOpen
  \bibfield  {author} {\bibinfo {author} {\bibfnamefont {A.~W.}\ \bibnamefont
  {Harrow}},\ }\href@noop {} {\bibfield  {journal} {\bibinfo  {journal} {arXiv
  preprint arXiv:1308.6595}\ } (\bibinfo {year} {2013})}\BibitemShut {NoStop}%
\bibitem [{\citenamefont {Peres}(1996)}]{Per96}%
  \BibitemOpen
  \bibfield  {author} {\bibinfo {author} {\bibfnamefont {A.}~\bibnamefont
  {Peres}},\ }\href {\doibase 10.1103/PhysRevLett.77.1413} {\bibfield
  {journal} {\bibinfo  {journal} {Phys. Rev. Lett.}\ }\textbf {\bibinfo
  {volume} {77}},\ \bibinfo {pages} {1413} (\bibinfo {year}
  {1996})}\BibitemShut {NoStop}%
\bibitem [{\citenamefont {Horodecki}\ \emph {et~al.}(1996)\citenamefont
  {Horodecki}, \citenamefont {Horodecki},\ and\ \citenamefont
  {Horodecki}}]{HHH96}%
  \BibitemOpen
  \bibfield  {author} {\bibinfo {author} {\bibfnamefont {M.}~\bibnamefont
  {Horodecki}}, \bibinfo {author} {\bibfnamefont {P.}~\bibnamefont
  {Horodecki}}, \ and\ \bibinfo {author} {\bibfnamefont {R.}~\bibnamefont
  {Horodecki}},\ }\href {\doibase 10.1016/S0375-9601(96)00706-2} {\bibfield
  {journal} {\bibinfo  {journal} {Phys. Lett. A}\ }\textbf {\bibinfo {volume}
  {223}},\ \bibinfo {pages} {1} (\bibinfo {year} {1996})}\BibitemShut {NoStop}%
\bibitem [{\citenamefont {Doherty}\ \emph {et~al.}(2004)\citenamefont
  {Doherty}, \citenamefont {Parrilo},\ and\ \citenamefont
  {Spedalieri}}]{DPS04}%
  \BibitemOpen
  \bibfield  {author} {\bibinfo {author} {\bibfnamefont {A.~C.}\ \bibnamefont
  {Doherty}}, \bibinfo {author} {\bibfnamefont {P.~A.}\ \bibnamefont
  {Parrilo}}, \ and\ \bibinfo {author} {\bibfnamefont {F.~M.}\ \bibnamefont
  {Spedalieri}},\ }\href {\doibase 10.1103/PhysRevA.69.022308} {\bibfield
  {journal} {\bibinfo  {journal} {Phys. Rev. A}\ }\textbf {\bibinfo {volume}
  {69}},\ \bibinfo {pages} {022308} (\bibinfo {year} {2004})}\BibitemShut
  {NoStop}%
\bibitem [{\citenamefont {Doherty}\ \emph {et~al.}(2005)\citenamefont
  {Doherty}, \citenamefont {Parrilo},\ and\ \citenamefont
  {Spedalieri}}]{DPS05}%
  \BibitemOpen
  \bibfield  {author} {\bibinfo {author} {\bibfnamefont {A.~C.}\ \bibnamefont
  {Doherty}}, \bibinfo {author} {\bibfnamefont {P.~A.}\ \bibnamefont
  {Parrilo}}, \ and\ \bibinfo {author} {\bibfnamefont {F.~M.}\ \bibnamefont
  {Spedalieri}},\ }\href {\doibase 10.1103/PhysRevA.71.032333} {\bibfield
  {journal} {\bibinfo  {journal} {Phys. Rev. A}\ }\textbf {\bibinfo {volume}
  {71}},\ \bibinfo {pages} {032333} (\bibinfo {year} {2005})}\BibitemShut
  {NoStop}%
\bibitem [{\citenamefont {Navascu\'es}\ \emph {et~al.}(2009)\citenamefont
  {Navascu\'es}, \citenamefont {Owari},\ and\ \citenamefont {Plenio}}]{NOP09}%
  \BibitemOpen
  \bibfield  {author} {\bibinfo {author} {\bibfnamefont {M.}~\bibnamefont
  {Navascu\'es}}, \bibinfo {author} {\bibfnamefont {M.}~\bibnamefont {Owari}},
  \ and\ \bibinfo {author} {\bibfnamefont {M.~B.}\ \bibnamefont {Plenio}},\
  }\href {\doibase 10.1103/PhysRevA.80.052306} {\bibfield  {journal} {\bibinfo
  {journal} {Phys. Rev. A}\ }\textbf {\bibinfo {volume} {80}},\ \bibinfo
  {pages} {052306} (\bibinfo {year} {2009})}\BibitemShut {NoStop}%
\bibitem [{\citenamefont {Brand\~ao}\ and\ \citenamefont
  {Christandl}(2012)}]{BC12}%
  \BibitemOpen
  \bibfield  {author} {\bibinfo {author} {\bibfnamefont {F.~G. S.~L.}\
  \bibnamefont {Brand\~ao}}\ and\ \bibinfo {author} {\bibfnamefont
  {M.}~\bibnamefont {Christandl}},\ }\href {\doibase
  10.1103/PhysRevLett.109.160502} {\bibfield  {journal} {\bibinfo  {journal}
  {Phys. Rev. Lett.}\ }\textbf {\bibinfo {volume} {109}},\ \bibinfo {pages}
  {160502} (\bibinfo {year} {2012})}\BibitemShut {NoStop}%
\bibitem [{\citenamefont {Myhr}(2011)}]{myhr2011symmetric}%
  \BibitemOpen
  \bibfield  {author} {\bibinfo {author} {\bibfnamefont {G.~O.}\ \bibnamefont
  {Myhr}},\ }\href@noop {} {\bibfield  {journal} {\bibinfo  {journal} {arXiv
  preprint arXiv:1103.0766}\ } (\bibinfo {year} {2011})}\BibitemShut {NoStop}%
\bibitem [{\citenamefont {Vandenberghe}\ and\ \citenamefont
  {Boyd}(1996)}]{VB96}%
  \BibitemOpen
  \bibfield  {author} {\bibinfo {author} {\bibfnamefont {L.}~\bibnamefont
  {Vandenberghe}}\ and\ \bibinfo {author} {\bibfnamefont {S.}~\bibnamefont
  {Boyd}},\ }\href {\doibase 10.1137/1038003} {\bibfield  {journal} {\bibinfo
  {journal} {SIAM Rev.}\ }\textbf {\bibinfo {volume} {38}},\ \bibinfo {pages}
  {49} (\bibinfo {year} {1996})}\BibitemShut {NoStop}%
\bibitem [{\citenamefont {Navascu{\'e}s}\ \emph {et~al.}(2009)\citenamefont
  {Navascu{\'e}s}, \citenamefont {Owari},\ and\ \citenamefont
  {Plenio}}]{navascues2009power}%
  \BibitemOpen
  \bibfield  {author} {\bibinfo {author} {\bibfnamefont {M.}~\bibnamefont
  {Navascu{\'e}s}}, \bibinfo {author} {\bibfnamefont {M.}~\bibnamefont
  {Owari}}, \ and\ \bibinfo {author} {\bibfnamefont {M.~B.}\ \bibnamefont
  {Plenio}},\ }\href {\doibase 10.1103/PhysRevA.80.052306} {\bibfield
  {journal} {\bibinfo  {journal} {Phys. Rev. A}\ }\textbf {\bibinfo {volume}
  {80}},\ \bibinfo {pages} {052306} (\bibinfo {year} {2009})}\BibitemShut
  {NoStop}%
\bibitem [{\citenamefont {Doherty}(2014)}]{doherty2014entanglement}%
  \BibitemOpen
  \bibfield  {author} {\bibinfo {author} {\bibfnamefont {A.~C.}\ \bibnamefont
  {Doherty}},\ }\href {\doibase 10.1088/1751-8113/47/42/424004} {\bibfield
  {journal} {\bibinfo  {journal} {J. Phys. A: Math. Gen.}\ }\textbf {\bibinfo
  {volume} {47}},\ \bibinfo {pages} {424004} (\bibinfo {year}
  {2014})}\BibitemShut {NoStop}%
\bibitem [{\citenamefont {Chiribella}(2011)}]{chiribella2010quantum}%
  \BibitemOpen
  \bibfield  {author} {\bibinfo {author} {\bibfnamefont {G.}~\bibnamefont
  {Chiribella}},\ }in\ \href {\doibase 10.1007/978-3-642-18073-6_2} {\emph
  {\bibinfo {booktitle} {Theory of Quantum Computation, Communication, and
  Cryptography}}},\ \bibinfo {series} {Lecture Notes in Computer Science},
  Vol.\ \bibinfo {volume} {6519},\ \bibinfo {editor} {edited by\ \bibinfo
  {editor} {\bibfnamefont {W.}~\bibnamefont {van Dam}}, \bibinfo {editor}
  {\bibfnamefont {V.}~\bibnamefont {Kendon}}, \ and\ \bibinfo {editor}
  {\bibfnamefont {S.}~\bibnamefont {Severini}}}\ (\bibinfo  {publisher}
  {Springer Berlin Heidelberg},\ \bibinfo {year} {2011})\ pp.\ \bibinfo {pages}
  {9--25}\BibitemShut {NoStop}%
\bibitem [{\citenamefont {Watrous}(2011)}]{Watrous2011}%
  \BibitemOpen
  \bibfield  {author} {\bibinfo {author} {\bibfnamefont {J.}~\bibnamefont
  {Watrous}},\ }\href@noop {} {\enquote {\bibinfo {title} {Theory of quantum
  information},}\ }\bibinfo {howpublished} {Lecture notes for CS766/QIC820,
  University of Waterloo} (\bibinfo {year} {2011}),\ \bibinfo {note} {the proof
  of this fact follows along similar lines of Lemma 22.1 (Lecture
  22).}\BibitemShut {Stop}%
\bibitem [{\citenamefont {Myhr}\ and\ \citenamefont
  {L\"utkenhaus}(2009)}]{ML09}%
  \BibitemOpen
  \bibfield  {author} {\bibinfo {author} {\bibfnamefont {G.~O.}\ \bibnamefont
  {Myhr}}\ and\ \bibinfo {author} {\bibfnamefont {N.}~\bibnamefont
  {L\"utkenhaus}},\ }\href {\doibase 10.1103/PhysRevA.79.062307} {\bibfield
  {journal} {\bibinfo  {journal} {Phys. Rev. A}\ }\textbf {\bibinfo {volume}
  {79}},\ \bibinfo {pages} {062307} (\bibinfo {year} {2009})}\BibitemShut
  {NoStop}%
\bibitem [{\citenamefont {Cerf}(2000)}]{Cerf2000}%
  \BibitemOpen
  \bibfield  {author} {\bibinfo {author} {\bibfnamefont {N.~J.}\ \bibnamefont
  {Cerf}},\ }\href {\doibase 10.1103/PhysRevLett.84.4497} {\bibfield  {journal}
  {\bibinfo  {journal} {Phys. Rev. Lett.}\ }\textbf {\bibinfo {volume} {84}},\
  \bibinfo {pages} {4497} (\bibinfo {year} {2000})}\BibitemShut {NoStop}%
\bibitem [{\citenamefont {Niu}\ and\ \citenamefont
  {Griffiths}(1998)}]{niu1998optimal}%
  \BibitemOpen
  \bibfield  {author} {\bibinfo {author} {\bibfnamefont {C.-S.}\ \bibnamefont
  {Niu}}\ and\ \bibinfo {author} {\bibfnamefont {R.~B.}\ \bibnamefont
  {Griffiths}},\ }\href {\doibase 10.1103/PhysRevA.58.4377} {\bibfield
  {journal} {\bibinfo  {journal} {Phys. Rev. A}\ }\textbf {\bibinfo {volume}
  {58}},\ \bibinfo {pages} {4377} (\bibinfo {year} {1998})}\BibitemShut
  {NoStop}%
\bibitem [{\citenamefont {Cubitt}\ \emph {et~al.}(2008)\citenamefont {Cubitt},
  \citenamefont {Ruskai},\ and\ \citenamefont {Smith}}]{cubitt2008structure}%
  \BibitemOpen
  \bibfield  {author} {\bibinfo {author} {\bibfnamefont {T.~S.}\ \bibnamefont
  {Cubitt}}, \bibinfo {author} {\bibfnamefont {M.~B.}\ \bibnamefont {Ruskai}},
  \ and\ \bibinfo {author} {\bibfnamefont {G.}~\bibnamefont {Smith}},\ }\href
  {\doibase 10.1063/1.2953685} {\bibfield  {journal} {\bibinfo  {journal} {J.
  Math. Phys.}\ }\textbf {\bibinfo {volume} {49}},\ \bibinfo {pages} {102104}
  (\bibinfo {year} {2008})}\BibitemShut {NoStop}%
\bibitem [{\citenamefont {Werner}(1989{\natexlab{a}})}]{Wer89}%
  \BibitemOpen
  \bibfield  {author} {\bibinfo {author} {\bibfnamefont {R.~F.}\ \bibnamefont
  {Werner}},\ }\href {\doibase 10.1103/PhysRevA.40.4277} {\bibfield  {journal}
  {\bibinfo  {journal} {Phys. Rev. A}\ }\textbf {\bibinfo {volume} {40}},\
  \bibinfo {pages} {4277} (\bibinfo {year} {1989}{\natexlab{a}})}\BibitemShut
  {NoStop}%
\bibitem [{\citenamefont {Werner}(1990)}]{werner1990remarks}%
  \BibitemOpen
  \bibfield  {author} {\bibinfo {author} {\bibfnamefont {R.~F.}\ \bibnamefont
  {Werner}},\ }\href {\doibase 10.1007/BF00429951} {\bibfield  {journal}
  {\bibinfo  {journal} {Lett. Math. Phys.}\ }\textbf {\bibinfo {volume} {19}},\
  \bibinfo {pages} {319} (\bibinfo {year} {1990})}\BibitemShut {NoStop}%
\bibitem [{\citenamefont {Johnson}\ and\ \citenamefont {Viola}(2013)}]{JV13}%
  \BibitemOpen
  \bibfield  {author} {\bibinfo {author} {\bibfnamefont {P.~D.}\ \bibnamefont
  {Johnson}}\ and\ \bibinfo {author} {\bibfnamefont {L.}~\bibnamefont
  {Viola}},\ }\href {\doibase 10.1103/PhysRevA.88.032323} {\bibfield  {journal}
  {\bibinfo  {journal} {Phys. Rev. A}\ }\textbf {\bibinfo {volume} {88}},\
  \bibinfo {pages} {032323} (\bibinfo {year} {2013})}\BibitemShut {NoStop}%
\bibitem [{\citenamefont {Zeng}\ \emph {et~al.}(2015)\citenamefont {Zeng},
  \citenamefont {Chen}, \citenamefont {Zhou},\ and\ \citenamefont
  {Wen}}]{ZCZW15}%
  \BibitemOpen
  \bibfield  {author} {\bibinfo {author} {\bibfnamefont {B.}~\bibnamefont
  {Zeng}}, \bibinfo {author} {\bibfnamefont {X.}~\bibnamefont {Chen}}, \bibinfo
  {author} {\bibfnamefont {D.-L.}\ \bibnamefont {Zhou}}, \ and\ \bibinfo
  {author} {\bibfnamefont {X.-G.}\ \bibnamefont {Wen}},\ }\href@noop {}
  {\enquote {\bibinfo {title} {Quantum information meets quantum matter -- from
  quantum entanglement to topological phase in many-body systems},}\ }\bibinfo
  {howpublished} {arXiv:1508.02595} (\bibinfo {year} {2015})\BibitemShut
  {NoStop}%
\bibitem [{\citenamefont {Werner}(1989{\natexlab{b}})}]{Werner1989}%
  \BibitemOpen
  \bibfield  {author} {\bibinfo {author} {\bibfnamefont {R.~F.}\ \bibnamefont
  {Werner}},\ }\href {\doibase 10.1007/BF00399761} {\bibfield  {journal}
  {\bibinfo  {journal} {Lett. Math. Phys.}\ }\textbf {\bibinfo {volume} {17}},\
  \bibinfo {pages} {359} (\bibinfo {year} {1989}{\natexlab{b}})}\BibitemShut
  {NoStop}%
\bibitem [{\citenamefont {Terhal}(2003)}]{Terhal2003}%
  \BibitemOpen
  \bibfield  {author} {\bibinfo {author} {\bibfnamefont {B.~M.}\ \bibnamefont
  {Terhal}},\ }\href {\doibase 10.1147/rd.481.0071} {\bibfield  {journal}
  {\bibinfo  {journal} {IBM J. Res. Dev.}\ }\textbf {\bibinfo {volume} {48}},\
  \bibinfo {pages} {71} (\bibinfo {year} {2003})}\BibitemShut {NoStop}%
\bibitem [{\citenamefont {Bae}\ and\ \citenamefont {Acin}(2006)}]{Bae2006}%
  \BibitemOpen
  \bibfield  {author} {\bibinfo {author} {\bibfnamefont {J.}~\bibnamefont
  {Bae}}\ and\ \bibinfo {author} {\bibfnamefont {A.}~\bibnamefont {Acin}},\
  }\href {\doibase 10.1103/PhysRevLett.97.030402} {\bibfield  {journal}
  {\bibinfo  {journal} {Phys. Rev. Lett.}\ }\textbf {\bibinfo {volume} {97}},\
  \bibinfo {pages} {030402} (\bibinfo {year} {2006})}\BibitemShut {NoStop}%
\bibitem [{\citenamefont {Osborne}\ and\ \citenamefont
  {Verstraete}(2006)}]{Osborne2006}%
  \BibitemOpen
  \bibfield  {author} {\bibinfo {author} {\bibfnamefont {T.~J.}\ \bibnamefont
  {Osborne}}\ and\ \bibinfo {author} {\bibfnamefont {F.}~\bibnamefont
  {Verstraete}},\ }\href {\doibase 10.1103/PhysRevLett.96.220503} {\bibfield
  {journal} {\bibinfo  {journal} {Phys. Rev. Lett.}\ }\textbf {\bibinfo
  {volume} {96}},\ \bibinfo {pages} {220503} (\bibinfo {year}
  {2006})}\BibitemShut {NoStop}%
\bibitem [{\citenamefont {Hill}\ and\ \citenamefont {Wootters}(1997)}]{HW97}%
  \BibitemOpen
  \bibfield  {author} {\bibinfo {author} {\bibfnamefont {S.}~\bibnamefont
  {Hill}}\ and\ \bibinfo {author} {\bibfnamefont {W.~K.}\ \bibnamefont
  {Wootters}},\ }\href {\doibase 10.1103/PhysRevLett.78.5022} {\bibfield
  {journal} {\bibinfo  {journal} {Phys. Rev. Lett.}\ }\textbf {\bibinfo
  {volume} {78}},\ \bibinfo {pages} {5022} (\bibinfo {year}
  {1997})}\BibitemShut {NoStop}%
\bibitem [{\citenamefont {Wootters}(1998)}]{Woo98}%
  \BibitemOpen
  \bibfield  {author} {\bibinfo {author} {\bibfnamefont {W.~K.}\ \bibnamefont
  {Wootters}},\ }\href {\doibase 10.1103/PhysRevLett.80.2245} {\bibfield
  {journal} {\bibinfo  {journal} {Phys. Rev. Lett.}\ }\textbf {\bibinfo
  {volume} {80}},\ \bibinfo {pages} {2245} (\bibinfo {year}
  {1998})}\BibitemShut {NoStop}%
\bibitem [{\citenamefont {Diaconis}\ and\ \citenamefont
  {Freedman}(1980)}]{diaconis1980finite}%
  \BibitemOpen
  \bibfield  {author} {\bibinfo {author} {\bibfnamefont {P.}~\bibnamefont
  {Diaconis}}\ and\ \bibinfo {author} {\bibfnamefont {D.}~\bibnamefont
  {Freedman}},\ }\href {\doibase 10.1214/aop/1176994663} {\bibfield  {journal}
  {\bibinfo  {journal} {Ann. Prob.}\ ,\ \bibinfo {pages} {745}} (\bibinfo
  {year} {1980})}\BibitemShut {NoStop}%
\bibitem [{\citenamefont {Horodecki}\ and\ \citenamefont
  {Horodecki}(1999)}]{HH99}%
  \BibitemOpen
  \bibfield  {author} {\bibinfo {author} {\bibfnamefont {M.}~\bibnamefont
  {Horodecki}}\ and\ \bibinfo {author} {\bibfnamefont {P.}~\bibnamefont
  {Horodecki}},\ }\href {\doibase 10.1103/PhysRevA.59.4206} {\bibfield
  {journal} {\bibinfo  {journal} {Phys. Rev. A}\ }\textbf {\bibinfo {volume}
  {59}},\ \bibinfo {pages} {4206} (\bibinfo {year} {1999})}\BibitemShut
  {NoStop}%
\bibitem [{\citenamefont {Cerf}\ \emph {et~al.}(1999)\citenamefont {Cerf},
  \citenamefont {Adami},\ and\ \citenamefont {Gingrich}}]{CAG99}%
  \BibitemOpen
  \bibfield  {author} {\bibinfo {author} {\bibfnamefont {N.~J.}\ \bibnamefont
  {Cerf}}, \bibinfo {author} {\bibfnamefont {C.}~\bibnamefont {Adami}}, \ and\
  \bibinfo {author} {\bibfnamefont {R.~M.}\ \bibnamefont {Gingrich}},\ }\href
  {\doibase 10.1103/PhysRevA.60.898} {\bibfield  {journal} {\bibinfo  {journal}
  {Phys. Rev. A}\ }\textbf {\bibinfo {volume} {60}},\ \bibinfo {pages} {898}
  (\bibinfo {year} {1999})}\BibitemShut {NoStop}%
\end{thebibliography}%

\end{document}